\documentclass[sigconf,dvipsnames]{acmart}

\let\sigproof\proof\let\proof\relax
\let\sigendproof\endproof\let\endproof\relax

\usepackage[utf8]{inputenc}

\usepackage{amsmath,amsthm,amsfonts,amscd}
\usepackage{multirow,booktabs}
\usepackage{faktor}
\usepackage{mathrsfs}
\usepackage{kantlipsum}
\usepackage{mathtools}

\usepackage{xspace}
\usepackage{multicol}
\usepackage[linesnumbered,vlined,boxruled,boxed,algo2e]{algorithm2e}
\usepackage{enumitem}
\usepackage{tikz}
\usepackage{bm}

\let\proof\sigproof
\let\endproof\sigendproof

\newtheoremstyle{sig}
  {}
  {}
  {\itshape}
  {}
  {\scshape}
  {.}
  {.5em}
  {#1 #2\thmnote{\quad(#3)}}

\numberwithin{equation}{section} 
\usepackage{tikz-cd} 
\tikzset{
  symbol/.style={
    draw=none,
    every to/.append style={
      edge node={node [sloped, allow upside down, auto=false]{$#1$}}}
  }
}
\usetikzlibrary{patterns,calc}
\usepackage{float} 
\usepackage[labelfont=bf,labelsep=space]{caption} 
\usepackage{empheq} 

\usepackage{subcaption}
\usepackage[textsize=tiny]{todonotes}

\newtheorem{lem}{Lemma}[section]
\newtheorem{prop}[lem]{Proposition}
\newtheorem{theo}[lem]{Theorem}
\newtheorem{cor}[lem]{Corollary}

\theoremstyle{definition}
\newtheorem{defi}[lem]{Definition}
\theoremstyle{remark}

\newtheorem{exam}[lem]{Example}

\newtheorem{remark}[lem]{Remark}
\newcommand{\eproof}{\hfill\qed}

\def\mcG{\mathcal{G}}

\def\mcP{\mathcal{P}}

\def\mcZ{\mathcal{Z}}

\def\bbE{\mathbb{E}}

\def\bbR{\mathbb{R}}
\def\bbZ{\mathbb{Z}}
\def\bbN{\mathbb{N}}

\def\bbC{\mathbb{C}}
\def\bbP{\mathbb{P}}

\def\bbD{\mathbb{D}}
\def\fkf{\mathfrak{f}}

\def\fkc{\mathfrak{c}}

\def\fkx{\mathfrak{x}}
\def\fky{\mathfrak{y}}
\def\fkz{\mathfrak{z}}
\def\sgn{\mathsf{sgn}}

\def\NN{\mathbb{N}\xspace}

\def\RR{\mathbb{R}\xspace}
\def\CC{\mathbb{C}\xspace}

\def\Pd{\mcP_{d}}

\def\Oh{\mathcal{O}}

\DeclareMathOperator{\cond}{\texttt{C}}
\DeclareMathOperator{\condR}{\texttt{C}_{\bbR}}

\def\uR{^{\mathbb{R}}}
\def\uC{^{\mathbb{C}}}
\def\enumber{\mathrm{e}}
\def\dist{\mathrm{dist}}

\DeclarePairedDelimiter\ceil{\lceil}{\rceil}

\DeclarePairedDelimiter\abs{\lvert}{\rvert}
\DeclarePairedDelimiter\norm{\lVert}{\rVert}

\newcommand{\OO}{\ensuremath{\mathcal{O}}\xspace}
\newcommand{\sO}{\ensuremath{\widetilde{\mathcal{O}}}\xspace}
\newcommand{\OB}{\ensuremath{\mathcal{O}_B}\xspace}
\newcommand{\sOB}{\ensuremath{\widetilde{\mathcal{O}}_B}\xspace}

\newcommand{\sturm}{\textsc{sturm}\xspace}

\newcommand{\descartes}{\textsc{descartes}\xspace}
\newcommand{\Descartes}{\textsc{descartes}\xspace}

\newcommand{\uObD}{{\overline{\mathcal{D}}}\xspace}
\newcommand{\dObD}{{\underline{\mathcal{D}}}\xspace}
\newcommand{\ObL}{\mathcal{L}\xspace}
\newcommand{\ObA}{\mathcal{A}\xspace}

\newcommand{\wid}{\mathtt{wid}\xspace}
\newcommand{\xmid}{\mathtt{mid}\xspace}
\newcommand{\img}{\mathtt{i}\xspace}

\newcommand{\var}{\textsc{var}\xspace}
\newcommand{\func}[1]{\textsc{#1}\xspace}

\makeatletter
\let\original@algocf@latexcaption\algocf@latexcaption
\long\def\algocf@latexcaption#1[#2]{%
  \@ifundefined{NR@gettitle}{%
    \def\@currentlabelname{#2}%
  }{%
    \NR@gettitle{#2}%
  }%
  \original@algocf@latexcaption{#1}[{#2}]%
}

\setcopyright{none}
\renewcommand\footnotetextcopyrightpermission[1]{} 
\copyrightyear{2022}
\acmYear{2022}
\acmConference[]{}{}{}
\acmPrice{15.00}
\acmDOI{10.1145/3373207.3404062}
\acmISBN{978-1-4503-7100-1/20/07}

\title[Beyond Worst-Case Analysis for Root Isolation Algorithms]{Beyond Worst-Case Analysis for Root Isolation Algorithms}
\usepackage{comment}

\begin{document}

\sloppy

\author{Alperen Ergur}
\email{alperen.ergur@utsa.edu}
\affiliation{%
  \institution{The Univ. of Texas at San Antonio}
  \streetaddress{}
  \city{}
  \country{TX, USA}
}

\author{Josu\'e Tonelli-Cueto}
\email{josue.tonelli.cueto@bizkaia.eu}
\affiliation{%
  \institution{Inria Paris \& IMJ-PRG}
  \streetaddress{}
  \city{}
  \country{France}
}

\author{Elias Tsigaridas}
\email{elias.tsigaridas@inria.fr}
\affiliation{%
  \institution{Inria Paris and Sorbonne Universit\'e}
  \streetaddress{}
  \city{}
  \country{France}
}

\begin{abstract}
Isolating the real roots of univariate polynomials is a fundamental problem in
symbolic computation and it is arguably one of the most important problems in
computational mathematics. The problem has a long history decorated with
numerous ingenious algorithms and furnishes an active area of research.
However, the worst-case analysis of root-finding algorithms does not correlate
with their practical performance. We develop a smoothed analysis framework for
polynomials with integer coefficients to bridge the gap between the complexity
estimates and the practical performance. 
In this setting, we derive that the expected bit complexity of 
\descartes solver to isolate the real roots of 
a polynomial, with coefficients uniformly distributed, is 
$\sOB(d^2 + d \tau)$,
where $d$ is the degree of the polynomial and $\tau$ the bitsize 
of the coefficients.
\end{abstract}




\keywords{
univariate polynomials,
root-finding,
Descartes solver,
condition-based complexity,
average complexity,
beyond worst-case analysis
}

\maketitle

\section{Introduction}
The interactions between the ways we design and the ways we analyze algorithms
are transformative on both ends: Unreasonably effective algorithms transform
our complexity analysis frameworks, where else the discovery of essential
complexity parameters transforms the ways we design algorithms. In numerical
computation, the use of the condition numbers illustrate emphatically this
phenomenon: condition numbers are  a way of explaining the success of certain
numerical algorithms\footnote{According to Wilkinson~\cite{wilkinson1971},
Turing~\cite{turing1948} introduced condition numbers to explain the practical
success of Gaussian elimination despite the existing worst-case analyses.}, a
guiding complexity parameter for the design of new algorithms, and a foundation
for average and smoothed analysis of numerical
algorithms~\cite{bcssbook,conditionbook}. In discrete computation, this
two-sided interaction between complexity analysis frameworks and algorithms'
design forms a dynamic and exciting area of current research
\cite{roughgarden2021book,downeyfellows2013} with a rich history rooted at the
beginnings of complexity theory~\cite[Ch.~18]{arorabarak2009}. Inspired by
these developments, we aim to take a first step for bringing different
modalities of algorithmic analysis into symbolic computation. To the best of
our knowledge, this large field
almost entirely relies on the worst-case analysis.

We consider one of the most basic problems in symbolic computation: computing
the roots of univariate polynomials. This is a singularly important problem with
applications in the whole range of computer science and engineering. It is
extensively studied from theoretical and practical perspectives for decades and
keeps attracting plenty of attention
\cite{McP-book-13,Pan-survey-21,Pan-survey-97,ept-crc-2012}.
We focus on the \emph{real root isolation} problem: to compute intervals with
rational endpoints that contain only one real root of the polynomial and each
real root is contained in an interval. Besides its countless direct
applications, this problem is omnipresent in symbolic computation; among its
numerous uses it stands out as a crucial subroutine for elimination based
multivariate polynomial systems solvers, e.g.,~\cite{ept-crc-2012}.

Despite the ubiquity of real root isolation in engineering and its relatively
long history in theoretical computer science, the state-of-the-art complexity
analysis falls short of providing guidance for practical computations. Pan's
algorithm~\cite{Pan02jsc} has the best worst-case complexity since nearly two
decades and is colloquially referred to as the ``optimal'' algorithm. However,
Pan's algorithm is rather sophisticated and has only a prototype implementation
in PARI/GP \cite{PARI2}.
In contrast, other algorithms with inferior worst-case complexity estimates
have excellent practical performance, e.g.,~
\cite{krs-for-real-16,htzekm-solve-09,et-tcs-2007}. In our view, this lasting
discrepancy between theoretical complexity analyses and practical performance
is related to the insistence on using the worst-case framework in the symbolic
computation community.
However, let us mention the exceptions of \cite{egt-issac-2010}, 
that provides estimates for the expected complexity of \textsc{sturm} algorithm for real solving,
\cite{et-tcs-2007}, that provides (conditional) expected case bounds for 
the continued fraction algorthm, 
and \cite{PanTsi-refine-2013} that considers the expected number
steps for the problem of real root refinement.

We demonstrate how average/smoothed analysis frameworks can help to
 predict the practical performance of symbolic real root isolation
  algorithms. In particular, we show that in our data model the
\descartes solver has a bit complexity quasi-linear in the input size, 
when we consider the dense representation.
This provides an explanation for the excellent practical performance of \descartes
that even
outperforms its numerical alternatives. See \S\ref{subsec:results1} for a simple
statement and \S\ref{subsec:results2} for the full technical statement.



\subsection{Synopsis of real root isolation algorithms}
\label{sec:synopsis}
We can (roughly) characterize the various algorithms for (real) root isolation
as numerical or symbolic algorithms;  the recent years there are also efforts
to combine the best of the two worlds.

The numerical algorithms are, in almost all the cases, iterative algorithms that
approximate all the roots (real and complex) of a polynomial up to any desired
precision. Their main common tool is (a variant of) a Newton operator. The
algorithm with the best worst-case complexity  due to Pan~\cite{Pan02jsc}  is
based on Sch\"onhage's splitting circle divide-and-conquer technique
\cite{Sch82}. It recursively factors the polynomial until we obtain linear
factors that approximate, up to any desired precision, all the roots of the
polynomial and it has nearly optimal arithmetic complexity. We can turn this
algorithm, and also any other numerical algorithm, to an exact one, by
approximating the roots up to the separation bound; that is
the minimum distance between the roots. In this way Pan obtained the record
worst case bit complexity bound $\sOB(d^2\tau)$ for a degree $d$ polynomial
with maximum coefficient bitsize $\tau$ \cite{Pan02jsc}; see also
\cite{Kirrinis-solve-98,msw-apan-15,becker2018near}.
Besides the algorithms already mentioned, there are also several seemingly
practically efficient numerical algorithms, e.g., \textsc{mpsolve}
\cite{mpsolve-theory} and \textsf{eigensolve} \cite{eigensolve}, that lack
convergence guarantees and/or precise bit complexity bounds.

Regarding symbolic algorithms, the majority is subdivision-based. These
algorithms mimic binary search. Given an initial interval that contains all (or
some) of the real roots, they repeatedly subdivide it until we obtain intervals
containing zero or one real root. 
Prominent representatives of this approach are \sturm and \descartes. \sturm
depends on Sturm sequences to count \emph{exactly} the number of distinct roots
in an interval, even when the polynomial is not square-free. Its complexity is
$\sOB(d^4 \tau^2)$ \cite{Dav:TR:85,Yap:SturmBound:05} and it is not so
efficient in practice; the
bottleneck seems to be the high cost of  computing the Sturm sequence.
\descartes is
based on Descartes' rule of signs to bound the number of real roots of a
polynomial in an interval. Its worst case complexity is
$\sOB(d^4 \tau^2)$ \cite{ESY:descartes}. Even though its worst case bound is
similar to \sturm, the \descartes solver has excellent practical performance
and it can routinely solve polynomials of degree several thousands
\cite{RouZim:solve:03,JohKraLynRicRus-issac-06,t-slv-16,htzekm-solve-09}.
There are also other algorithms based on the continued fraction expansion of
the real numbers \cite{sharma-tcs-2008,et-tcs-2007}
and on point-wise evaluation \cite{BurrKrahmer-eval-12,sy-simple-11}.

Let us also mention the bitstream version of \descartes
\cite{ekkmsw-bitstream-05},
where we assume that there is an oracle that for each coefficient of the
polynomial returns an approximation to any absolute error.
This approach, by also incorporating several tools from numerical algorithms,
leads to improved variants of \descartes \cite{sm-anewdsc}.
In the end, this variant yields the record worst case complexity bounds and
efficient implementation~\cite{krs-for-real-16} especially
when there are clusters of roots.
Even more, there is also a subdivision algorithm  \cite{becker2018near}
that applies several improvements to the modified Weyl algorithm by Pan
\cite{Pan-Weyl-00}
and achieves the (record) complexity bound $\sOB(d^3 + d^2 \tau)$.

\subsection{Warm-up: A simple form of the main result}\label{subsec:results1}
The main complexity parameters for univariate polynomials with integer (or
rational) coefficients is the degree $d$ and the bitsize $\tau$; the latter
refers to the maximum bitsize of the coefficients. We aim for a data model that
resembles a ``typical'' polynomial with exact coefficients. The first natural
candidate is the following: fix a bitsize $\tau$, let
$\fkc_0,\fkc_1,\ldots,\fkc_d$
be independent copies of the uniformly distributed integer in
$[-2^{\tau},2^{\tau}] \cap \mathbb{Z}$, and let
$
  \fkf = \sum\nolimits_{i=0}^d \fkc_i X^{i}
$
which we call the \emph{uniform random bit polynomial with bitsize $\tau$}. 
Recall that $\OO$, resp. $\OB$, denote the arithmetic, resp.  bit,
complexity and that we use $\sO$, resp. $\sOB$, to ignore (poly-)logarithmic
factors of $d$. For uniform random bit polynomials, our result has the
following form.
\begin{theo}\label{theo:main1}
For a degree $d$ uniform random bit polynomial $\fkf$ with bit size
$\tau(\fkf)$,  \descartes solver isolates the real roots of $\fkf$ in  expected
time $\sOB(d \, \tau+d^2)$.
\end{theo}


Notice that the expected time complexity of \descartes solver in this simple
model is better by a factor of $d$ than the record worst-case complexity bound
of Pan's algorithm.

\subsection{Statement of main results in full detail}\label{subsec:results2}

We develop a general model of randomness that provides a smoothed analysis
framework on polynomials with integer coefficients.
\begin{defi}
Let $d\in\bbN$. A \emph{random bit polynomial with degree $d$} is a random
polynomial
$
\fkf:=\sum\nolimits_{i=0}^d \fkc_i X^{i} ,
$
where the $\fkc_i$ are independent discrete random variables with values in
$\bbZ$.
Then,
\begin{enumerate}
  \item the \emph{bitsize of $\fkf$}, $\tau(\fkf)$, is the minimum integer
    $\tau$ so that for all $i \in \{ 0,1,2,\ldots,d \}$,
    $
    \bbP(| \fkc_i| \leq 2^\tau)=1.
    $
  \item the \emph{weight of $\fkf$}, $w(\fkf)$, is the maximum probability that
    $\fkc_0$, $\fkc_1$, $\fkc_{d-1}$, and $\fkc_d$ can take a value, i.e.,
    \[
    w(\fkf):=\max\{\bbP(\fkc_i=k) \mid  i\in\{0,1,d-1,d\},\,k\in\bbR\}.
    \]
\end{enumerate}
\end{defi}
\vspace{0.05in}

\begin{remark}
Note that we only impose restrictions on the size of the probabilities of the
coefficients of $1$, $X$, $X^{d-1}$ and $X^d$. This might look odd at the first
sight. We set our randomness model this way to be able to consider the most
flexible data-model that can be handled by our proof techniques. We provide
examples below to justify this technical assumption.

\end{remark}

\begin{exam} \label{ex:uniform}
The uniform random bit polynomial of bitsize $\tau$ we introduced  is the main
example of a random bit polynomial $\fkf$. Note that in this case we have
$w(\fkf) = \frac{1}{1+2^{\tau+1}}$ and $\tau(\fkf)=\tau$.
\end{exam}

As we will see in the examples below, our randomness model is very flexible.
However, this flexibility comes at a cost. In
principle, we could have $w(\fkf)=1$ which would make our randomness model
equivalent to the worst-case model. To control the effect of large $w(\fkf)$ we
introduce the following quantity, which measures how far we are from a uniform
random bit polynomial.

\begin{defi}
  \label{def:uniformity}
  The \emph{uniformity} of a random bit polynomial $\fkf$ is
  \[
    u(\fkf):=\ln \left( w(\fkf) (1+2^{\tau(\fkf)+1}) \right) .
  \]
\end{defi}
\begin{remark}
Note that $u(\fkf)=0$ if and only if the coefficients of $1$, $X$, $X^{d-1}$
and $X^d$ in $\fkf$ are uniformly distributed in $[-2^{\tau}, 2^{\tau}] \cap
\mathbb{Z}$.
\end{remark}

The following three examples illustrate the flexibility of our random model by
specifying the support, the sign of the coefficients, and their exact bitsize.
Although we specify them separately, any combination of specifications is also
possible.

\begin{exam}[Support]\label{ex:spec1}
Let $A\subseteq \{0,1,\ldots,d-1,d\}$ with $0,1,d-1,d\in A$. Then
$\fkf:=\sum_{i\in A}\fkc_iX^i$,
where the $\fkc_i$'s are independent and uniformly distributed in
$[-2^{\tau},2^\tau]$ is a random bit polynomial with $u(\fkf)=0$ and
$\tau(\fkf)=\tau$.
\end{exam}

\begin{exam}[Sign of the coefficients]\label{ex:spec2}
Let $s\in\{-1,+1\}^{d+1}$. The random polynomial $\fkf:=\sum_{i=0}^d\fkc_iX^i$,
where the  $\fkc_i$'s are independent and uniformly distributed in
$s_i([1,2^{\tau}]\cap\bbN)$, is a random bit polynomial with $u(\fkf)\leq
\ln(3)$ and $\tau(\fkf)=\tau$.
\end{exam}

\begin{exam}[Exact bitsize]\label{ex:spec3}
Let $\fkf:=\sum_{i=0}^d\fkc_iX^i$ be the random polynomial, where the
$\fkc_i$'s are independent random integers of exact bitsize $\tau$, i.e.,
$\fkc_i$ is uniformly distributed in
$\bbZ\cap([-2^{\tau}+1,-2^{\tau-1}]\cup[2^{\tau-1},2^{\tau}-1])$. Then $\fkf$
is a random bit polynomial with $u(\fkf)\leq \ln(3)$ and $\tau(\fkf)=\tau$.
\end{exam}


We consider a \emph{smoothed random model} for polynomials, where a
deterministic polynomial is perturbed by a random one. In this way our
random bit polynomial model includes smoothed analysis over integer
coefficients as a special case.

\begin{exam}[Smoothed analysis]\label{ex:smoothed}
Let $f\in \Pd$ be a fixed integer polynomial with coefficients in
$[-2^\tau,2^\tau]$, $\sigma\in\bbZ\setminus \{0\}$ and $\fkf\in\Pd$ a random
bit polynomial. Then
$
\fkf_\sigma:=f+\sigma\fkf 
$
is a random bit-polynomial with bitsize
$
\tau(\fkf_\sigma)\leq \max\{\tau,\tau(\fkf)+\tau(\sigma)\}+1,
$
where $\tau(a)$ denotes the bitsize of $a$, and uniformity
$
u(\fkf_\sigma)\leq 1+\max\{\tau-\tau(\fkf),\tau(\sigma)\}+u(\fkf).
$
By combining the smoothed random model with the previous examples, we can
obtain structured perturbations.
\end{exam}

Our main result is the following:
\begin{theo}\label{theo:main2}
  Let $\fkf$ be random bit polynomial, of degree $d$, bitsize $\tau(\fkf)$, and
  uniformity parameter $u(\fkf)$, such that
$\tau(\fkf) = \Omega( \lg{d}+ u(\fkf))$, then  \descartes solver isolates the
real roots of $\fkf$ in  expected time $\sOB(d \, \tau \, (1 + u(\fkf))^3
+d^2\, (1 + u(\fkf))^4)$.
\end{theo}

\begin{remark}
Note that if $\fkf$ is not square-free, \descartes will compute its square-free
part and then proceed as usual. The probabilistic complexity estimate covers
this case.
\end{remark}


\begin{remark}
One might further optimize the probabilistic estimates present in Section
\ref{subsec:prob} by employing strong tools from Littlewood-Offord
theory~\cite{rudelsonvershynin2008}. However, the complexity analysis depends
on the random variables in a logarithmic scale and so further improvements on
probabilistic estimates will not make any essential improvement on our main
result. Therefore, we prefer to use more transparent proofs with slightly less
optimal dependency on the uniformity parameter $u(\fkf)$.
\end{remark}



\subsection{Overview of main ideas}

The important quantities in analyzing \descartes are the \emph{separation
bound} and the number of complex roots nearby the real axis.

The separation bound is the minimum distance between
the distinct (complex) roots of a polynomial.
\cite{emt-dmm-j-19}.
This quantity controls the depth of the subdivision tree of \descartes. To
estimate this quantity we use condition
numbers~\cite{bcssbook,dedieubook,conditionbook}, following
\cite{TCTcubeI-journal}. In short, we use condition numbers to obtain an
instance-based estimate for the depth of the subdivision tree of \descartes.
Even though the \descartes algorithm isolates the real roots, complex roots
near the real axis control the width of the subdivision tree. This fact follows
from the work of Obreshkoff~\cite{Obreshkoff-book}, see also
\cite{KM-newDesc-06}. To estimate the number of roots in the Obreshkoff areas
we use complex analytic techniques. In short, we bound the number of complex
roots in a certain region to obtain an instance-based
estimate for the width of the subdivision tree of \descartes.
Overall, by controlling both the depth---through the condition number---and the
width---through the number of complex roots---we estimate the size of the
subsdivision tree of \descartes and so its bit complexity.

Finally, we perform the expected/smoothed analysis of \descartes by performing
probabilistic analyses of the number of complex roots and the condition number.
Expected/smoothed analysis results in computational algebraic geometry are rare
and mostly restricted to continuous random variables, with few exceptions
\cite{castromontanapardosanmartin2002};
see also \cite{PanTsi-refine-2013,egt-issac-2010,et-tcs-2007}.
To the best of our knowledge,  we
present the first known result for the expected complexity of root finding for
random polynomials with integer coefficients. Our results rely on the strong
toolbox developed by Rudelson, Vershynin, and others in random matrix
theory~\cite{rudelsonvershynin2015,livshytspaourispivovarov2016}.

\paragraph*{Organization.}
We deal with condition numbers, separation bounds and their probabilistic
estimates in Section~\ref{sec:cond-sep-prob}, we deal with the estimates of the
number of complex roots in Section~\ref{sec:complexroots}, and we show how
these quantities control the complexity of \descartes obtaining the final
complexity estimate in Section~\ref{sec:Descartes}.

\paragraph*{Notation.}
We denote by $\OO$, resp. $\OB$, the arithmetic, resp.  bit,
complexity and we use $\sO$, resp. $\sOB$, to ignore (poly-)logarithmic factors
of $d$. We denote by $\Pd$ the space of univariate polynomials of degree at
most $d$ with real coefficients and by $\Pd^\bbZ$ the subset of integer
polynomial. If $f = \sum_{k=0}^df_k X^k \in \Pd^\bbZ$, then the bitsize of $f$
is the maximum bitsize of its
coefficients. The set of complex roots of $f$ is $\mcZ(f)$.
We denote by $\var(f)$ the number of sign changes in the
coefficient list.
The \emph{separation bound} of $f$, $\Delta(f)$ or $\Delta$ if $f$ is clear
from the context,
is the minimum distance between the roots of $f$, see
\cite{emt-dmm-j-19,ep-dmm-17,Dav:TR:85}.
We denote by $\bbD$ the unit disc in the complex plane, by $\bbD(x,r)$ the disk
$x+r\bbD$, and by $I$ the interval $[-1,1]$. For a real interval $J = (a, b)$,
we consider $\xmid(J): = \tfrac{a+b}{2}$ and $\wid(J):= b-a$. For a
$n \in \NN$, we use $[n]$ to signify the set $\{1, \dots, n\}$ and
$\mu(n) = \OB(n\lg n)$ for the complexity of multiplying two integers of
bitsize $n$, where $\lg$ is the logarithm with base 2.

\paragraph*{Acknowledgements.} 
J.T-C. is supported by a postdoctoral fellowship of the 2020
``Interaction'' program of the Fondation Sciences Mathématiques de Paris. He is
grateful to Evgenia Lagoda for moral support and Gato Suchen for useful
suggestions regarding
Proposition~\ref{prop:reallinearprojectiondiscreterandomvector}.
A.E. is supported by NSF CCF 2110075, J.T-C. and E.T. 
are partially
supported by ANR
JCJC GALOP (ANR-17-CE40-0009).

\section{Condition numbers, separation bounds, and randomness}
\label{sec:cond-sep-prob}

We use various condition numbers for univariate polynomials
from~\cite{TCTcubeI-journal}, cf.~\cite{TCTcubeI}, to control the separation
bound of random polynomials. However, our
probabilistic analysis differs from~\cite{TCTcubeI-journal}
because we consider discrete random coefficients. 

\subsection[Condition numbers for univariate polynomials]{Condition numbers for
univ. polynomials}

The \emph{local condition number of $f\in\Pd$ at
  $z\in \bbD$}~\cite{TCTcubeI-journal} is
\begin{equation}
	\label{eq:cond-in-D}
    \cond(f,z):=\frac{\|f\|_1}{\max\{|f(z)|,|f'(z)|/d\}},
\end{equation}
where $\|f\|_1:=\sum |f_a|$ is the \emph{1-norm} of $f$. 
important for obtaining bit complexity results. The same definition using the
$\ell_2$-norm is standard in numerical analysis literature, e.g.,
\cite{Higham-book-02}.

We also define the \emph{(real) global condition number of $f$} as
\begin{equation}
  \label{eq:condR}
  \condR(f):=\max_{x\in I}\cond(f,x).
\end{equation}

We note that as $\condR(f)$ becomes bigger, $f$ is closer to have 
a singular real zero inside $I$. This can be made precise through the so-called
condition number theorem (see~\cite[Theorem~4.4]{TCTcubeI-journal}). There are
many interesting properties of $\condR(f)$, but let us state the only one we
will use---see \cite[Theorem~4.2]{TCTcubeI-journal} for more.

\begin{theo}[2nd Lipschitz
property]\cite{TCTcubeI-journal}\label{theo:propconditionnumber}
Let $f\in\Pd$. The map $\bbD\ni z\mapsto 1/\cond(f,z)\in [0,1]$ is well-defined
and $d$-Lipschitz.\eproof
\end{theo}

\subsection{Condition-based estimates for separation}

The quantity that follows is the separation bound of polynomials and polynomial
systems, e.g., \cite{emt-dmm-j-19}, suitably adjusted in our setting. This
quantity and its condition-based estimate below will play a fundamental role in
our complexity estimates.

\begin{defi}\label{defi:realseparation}
For $\varepsilon\in\left[0,\frac{1}{d}\right)$ we set $I_\varepsilon:=\{z\in
\bbC\mid \dist(z,I)\leq \varepsilon\}$. If $f \in \Pd$, then the
\emph{$\varepsilon$-real separation of $f$}, $\Delta\uR_{\varepsilon}(f)$, is
    \[
\Delta_{\varepsilon}\uR(f):=\min\left\{\left|\zeta-\tilde{\zeta}\right|\mid
\zeta,\tilde{\zeta}\in I_\varepsilon ,\,f\left(\zeta\right)
    =f(\tilde{\zeta})=0\right\},
  \]
    if $f$ has no double roots in $I_\varepsilon$, 
    and $\Delta\uR_{\varepsilon}(f):=0$  otherwise.
\end{defi}

\begin{theo}[{\cite[Theorem 6.3]{TCTcubeI-journal}}]
\label{theo:condbasedseparation}
Let $f\in\Pd$ and assume   $\varepsilon\in\left[0, \frac{1}{\enumber
d\condR(f)}\right)$, then $\Delta_{\varepsilon}\uR(f)\geq
\frac{1}{12d\condR(f)}.$\eproof
\end{theo}

\subsection{Probabilistic bounds for condition numbers}\label{subsec:prob}

In this section we present our probabilistic framework. The main technical tools
are the anti-concentration results by Rudelson and
Vershynin~\cite{rudelsonvershynin2015}. We do not apply these results as a
black box, but we develop suitable variants for our setting
(Proposition~\ref{prop:reallinearprojectiondiscreterandomvector}).

\begin{theo}\label{theo:reallocalconditionnumberprob}
Let $\fkf \in \Pd^{\bbZ}$ be a  random bit polynomial and $x\in I$. Then, for
$t\leq 2^{\tau(\fkf)}$,
$
\bbP(\cond(\fkf,x)\geq t)\leq 16\,d^3\enumber^{2u(\fkf)}\frac{1}{t^2}
$.
\end{theo}
\begin{theo}\label{theo:realglobalconditionnumberprob}
Let $\fkf\in\Pd^{\bbZ}$ be a random bit polynomial. Then, for $t\leq
2^{\tau(\fkf)+1}$,
\[
\bbP(\cond_\bbR(\fkf)\geq t)\leq 32\, d^4 \enumber^{2u(\fkf)}\frac{1}{t}.
\]
\end{theo}

The following corollary looks somewhat different than
Thm.~\ref{theo:reallocalconditionnumberprob} and
Thm.~\ref{theo:realglobalconditionnumberprob}, but it has the same essence.
Unlike the continuous case, in the discrete case we have a worst-case estimate
that we can exploit to bound when the condition number is too large.

\begin{cor}\label{cor:realglobalexpectations}
Let $\fkf\in\Pd^{\bbZ}$ be a random bit polynomial, $\ell\in\bbN$ and $c\geq
1$. If $\tau(\fkf)\geq 4\ln(\enumber d)+2u(\fkf)$, then $\left(\bbE_\fkf
\left(\min\{\ln\condR(\fkf),c\}\right)^\ell\right)^{\frac{1}{\ell}}$ is at most
\[(\ell+1)
\left(4\ln(ed)+2u(\fkf)\right)+\Big(\frac{16\,d^4\enumber^{u(\fkf)}}{2^{\tau
(\fkf)}}\Big)^{\frac{1}{\ell}}c.\]
In particular, if $\tau(\fkf)\geq 4\ln(ed)+2u(\fkf)+2\ell\ln c$, then
\[\left(\bbE_\fkf
\left(\min\{\ln\condR(\fkf),c\}\right)^\ell\right)^{\frac{1}{\ell}}\leq
\Oh(\ell(\ln d+u(\fkf))).\]
\end{cor}
We would like to understand the limitations of the two theorems and the
corollary above. First, note that
Theorem~\ref{theo:reallocalconditionnumberprob} is meaningful when
$
\tau(\fkf)\geq 2+\frac{3}{2}\lg(d) +2u(\fkf)
$
and Theorem~\ref{theo:realglobalconditionnumberprob} is meaningful when
$
   \tau(\fkf)\geq 5 + 4\lg(2)+ 3u(\fkf)
$.
Intuitively, the randomness model needs some wiggling room to differ from the
worst-case analysis. In our case this translates to assume that the bit-size
$\tau(\fkf)$ is bigger than (roughly) $\lg(d) + u(\fkf)$. This is a reasonable
assumption 
because for most cases of
interest, $u(\fkf)$ is bounded above by a constant. Thus, the second condition
in Corollary~\ref{cor:realglobalexpectations} becomes
\[
\tau(\fkf)=\Omega(\ell\lg(d)+\lg(c)).
\]
Moreover, in the case of application of
Corollary~\ref{cor:realglobalexpectations}, we will have $c=d^{\Oh(1)}$. Hence
we are only imposing that the bit-size
$\tau(\fkf)$ is lower bounded by (roughly) $\ln d$, which is not uncommon in
practice.

For proving the above results, we need the following proposition. Recall that
for $A\in\bbR^{k\times N}$,  $\|A\|_{\infty,\infty}:=\sup_{v\neq
0}\frac{\|Av\|_\infty}{\|v\|_\infty}=\max_{i\in k}\|A^i\|_1$, where $A^i$ is
the $i$-th row of $A$.

\begin{prop}\label{prop:reallinearprojectiondiscreterandomvector}
Let $\fkx\in\bbZ^N$ be a random vector with independent coordinates. Assume
that there is $w>0$ so that for all $i$ and $x\in\bbZ$,
$\bbP(\fkx_i=x)\leq w$.
Then for every linear map $A\in\bbR^{k\times N}$, $b\in \bbR^k$ and
$\varepsilon\in[\|A\|_{\infty,\infty},\infty)$,
\[
\bbP(\|A\fkx+b\|_\infty\leq \varepsilon )\leq
2\frac{(2\sqrt{2}w\varepsilon)^k}{\sqrt{\det AA^*}} .
\]
\end{prop}
\begin{proof}[Proof of Theorem~\ref{theo:reallocalconditionnumberprob}]
$\bbP(\cond(\fkf,x)\geq t)$ equals
\[
\sum_{a_2,\ldots,a_{d-2}}\bbP(\cond(\fkf,x)\geq t\mid
\fkc_2=a_2,\ldots,\fkc_{d-2}=a_{d-2})\prod_{i=2}^{d-2}\bbP(\fkc_i=a_i).
\]
where $\fkf=\sum_{k=0}^d\fkc_kX^k$. 
So it is enough to prove the bound for a random bit polynomial $\fkf$ of the
form
\[
\fkf=\fkc_0+\fkc_1X+\sum\nolimits_{k=2}^{d-2}a_kX^k+\fkc_{d-1}X^{d-1}+\fkc_dX^d
,
\]
where $a_2,\ldots,a_{d-2}\in \bbZ\cap [-2^\tau,2^\tau]$ are arbitrary fixed
integers.

Let $\Pd(a_2,\ldots,a_{d-2})$ be the affine subspace of $\Pd$ given by
$f_k=a_k$ for $k\in\{2,\ldots,d-2\}$. And let
\[
f\mapsto Af+b
\]
be the affine mapping that maps $f\in \Pd(a_2,\ldots,a_{d-2})$ to
$(f(x),f'(x)/d)\in\bbR^2$. In the coordinates we are working on (those of the
base $\{1,X,X^{d-1},X^d\}$), $A$ has the form
\[
\begin{pmatrix}
1&x&x^{d-1}&x^d\\
0&1/d&(1-1/d)x^{d-2}&x^{d-1}
\end{pmatrix}.
\]
So, by an elementary estimation we have $\|A\|_{\infty,\infty}\leq d+1$, and as
a direct result of Cauchy-Binet formula we have $\sqrt{\det AA^*}\geq 1/d$.
Now, since $\|\fkf\|_1\leq (d+1)2^{\tau(\fkf)}$, we have that
$
    \bbP(\cond(\fkf,x)\geq t)=\bbP(\|A\fkf+b\|_\infty\leq \|\fkf\|_1/t)
    \leq \bbP(\|A\fkf+b\|_\infty\leq (d+1)2^{\tau(\fkf)}/t)$.

To be able to use
Proposition~\ref{prop:reallinearprojectiondiscreterandomvector}, we need to
assume  $\frac{(d+1)2^{\tau(\fkf)}}{t} \geq d+1 \geq \norm{A}_{\infty,\infty}$.
Then, for $t \leq 2^{\tau(\fkf)}$,
Proposition~\ref{prop:reallinearprojectiondiscreterandomvector} implies
\[ \bbP(\cond(\fkf,x)\geq t) \leq 16 d (d+1)^2 (w 2^{\tau(\fkf)}/t)^2 . \]
Thus the proof is completed by the definition of $u(\fkf)$. 
\end{proof}
\begin{proof}[Proof of Theorem~\ref{theo:realglobalconditionnumberprob}]
We will use a covering/union bound argument. For any finite set
$\mcG\subset[-1,1]$ such that $\{[x-\delta,x+\delta]\mid x\in\mcG\}$ covers
$[-1,1]$, using the 2nd Lipschitz property
(Theorem~\ref{theo:propconditionnumber}), we have
$1/\max_{x\in\mcG}\cond(f,x)\leq 1/\condR(f)+d\delta$. Let $\delta=1/dt$, then
$
\bbP(\condR(\fkf)\geq t)\leq \bbP\left(\max_{x\in\mcG}\condR(\fkf,x)\geq
t/2\right)\leq \#\mcG\,\max_{x\in [-1,1]}\bbP(\cond(\fkf,x)\geq t/2).
$
We can construct such $\mcG$ such that $\#\mcG\leq 2dt$. Hence the claim
follows from Theorem~\ref{theo:reallocalconditionnumberprob}.
\end{proof}
\begin{proof}[Proof of Corollary~\ref{cor:realglobalexpectations}]
Let
\[
U:=\ln(32\,d^4\enumber^{2u(\fkf)})\leq 4\ln(ed)+2u(\fkf)\text{ and }
V:=\ln(2^{\tau(\fkf)+1}) .
\]
By assumption, $U\leq V$ and $U>1$, since $u(\fkf)\geq 0$.
So without loss of generality, we  assume
$
0<U<V<c.
$
If $c\leq V$, then similar arguments imply 
that the claimed upper bound still holds.
Thus 
\[
\bbE_\fkf \left(\min\{\ln\condR(\fkf),c\}\right)^\ell=\int_{0}^c \ell
s^{\ell-1}\bbP(\min\{\ln\condR(\fkf),c\}\geq s)\,\mathrm{d}s.
\]
We divide the integral into three summands using the intervals $[0,U]$, $[U,V]$
and $[V,c]$.

In $[0,U]$, we have that $\bbP(\min\{\ln\condR(\fkf),c\}\geq s)\leq 1$, and so
\[
\int_{0}^U \ell s^{\ell-1}\bbP(\min\{\ln\condR(\fkf),c\}\geq
s)\,\mathrm{d}s\leq U^\ell .
\]
In $[U,V]$, by Theorem~\ref{theo:realglobalconditionnumberprob} we have that
\[
\bbP(\min\{\ln\condR(\fkf),c\}\geq s)\leq \bbP(\ln\condR(\fkf)\geq s)\leq
e^{U-s},
\]
and so $\int_{U}^V \ell s^{\ell-1}\bbP(\min\{\ln\condR(\fkf),c\}\geq
s)\,\mathrm{d}s$ is bounded by $\int_{U}^V \ell
s^{\ell-1}e^{U-s}\,\mathrm{d}s$. By performing a change of variables and
extending the domain, we get $\int_{0}^{\infty} \ell
(s+U)^{\ell-1}e^{-s}\,\mathrm{d}s$. The latter, expanding the binomial
$(s+U)^{\ell-1}$ and using that $\Gamma(k+1)=k!$, is bounded by $\ell
\sum_{k=0}^{\ell-1}\binom{\ell-1}{k}k!U^{\ell-1-k}$.
Hence, as $\binom{\ell-1}{k}!\leq \ell^{\ell-1}$, we get
\[
\int_{U}^V \ell s^{\ell-1}\bbP(\min\{\ln\condR(\fkf),c\}\geq
s)\,\mathrm{d}s\leq \ell^\ell U^{\ell-1}.
\]

In $[V,c]$, we have that
\[
\bbP(\min\{\ln\condR(\fkf),c\}\geq s)
\leq \bbP(\ln\condR(\fkf)\geq V)\leq e^{U-V}.
\]
Therefore, since $e^{U-V}\int_{V}^c \ell s^{\ell-1}\,\mathrm{d}s\leq
e^{U-V}\int_{0}^c \ell s^{\ell-1}\,\mathrm{d}s$,
\[
\int_{V}^c \ell s^{\ell-1}\bbP(\min\{\ln\condR(\fkf),c\}\geq
s)\,\mathrm{d}s\leq e^{U-V}c^{\ell} .
\]
%
To obtain the final estimate, we add the three upper bounds obtaining the uper
bound
$
U^\ell+\ell^\ell U^{\ell-1}+e^{U-V}c^\ell.
$
After substituting the values of $U$ and $V$ and some easy estimations, we
conclude.
\end{proof}

\begin{proof}[Proof of
Proposition~\ref{prop:reallinearprojectiondiscreterandomvector}]
Let $\fky\in\bbR^N$ be such that the $\fky_i$ are independent and uniformly
distributed in $(-1/2,1/2)$. Now, a simple computation shows that $\fkx+\fky$
is absolutely continuous and each component has density given by
\[
\delta_{\fkx_i+\fky_i}(t)
=\sum\nolimits_{s\in \bbZ}\bbP(\fkx_i=s)\delta_{\fky_i}(t-s).
\]
Thus each component of $\fkx+\fky$ has density bounded by $w$.
We have
\[
\bbP(\|A\fkx+b\|_\infty\leq \varepsilon)\leq \bbP(\|A(\fkx+\fky)+b\|_\infty\leq
2\varepsilon)/\bbP(\|A\fky\|_\infty\leq \varepsilon),
\]
since $\fkx$ and $\fky$ are independent, and by the triangle inequality.

On the one hand, we apply \cite[Proposition~5.2]{TCTcubeI} (which is nothing
more than~\cite[Theorem~1.1]{rudelsonvershynin2015} with the explicit constants
of~\cite{livshytspaourispivovarov2016}). The latter states that for a random
vector $\fkz\in\bbR^N$ with independent coordinates with density bounded by
$\rho$ and $A\in \bbR^{k\times N}$, we have that $A\fkz$ has density bounded by
$(\sqrt{2}\rho)^k/\sqrt{\det AA^*}$. Thus
\[
\bbP(\|A(\fkx+\fky)+b\|_\infty\leq 2\varepsilon)\leq
(2\sqrt{2}w\varepsilon)^k/\sqrt{\det AA^*}.
\]
On the other hand,
\[
\bbP(\|A\fky\|_\infty\leq \varepsilon)=1-\bbP(\|A\fky\|_\infty\geq
\varepsilon)\geq 1-\bbE\|A\fky\|_\infty/\varepsilon .
\]
by Markov's inequality.
Now, by our assumption on $\varepsilon$, we only need to show that
$\bbE\|A\fky\|_\infty\leq \|A\|_{\infty,\infty}/2$.

By Jensen's inequality,
\[
\bbE\|A\fky\|_\infty=\bbE\lim_{\ell\to\infty}\|A\fky\|_{2\ell}\leq
\lim_{\ell\to\infty}\left(\bbE\|A\fky\|_{2\ell}^{2\ell}\right)^{\frac{1}{2\ell}
} .
\]
Expanding the interior and computing the moments of $\fky$, we obtain
\[
\bbE\|A\fky\|_\infty\leq
\lim_{\ell\to\infty}\left(\sum_{i=1}^k\sum_{|\alpha|=\ell}\binom{2\ell}{2\alpha
}\prod_{j=1}^n\left(A_{i,j}^{2\alpha_j}(1/2)^{2\alpha_j}/(2\alpha_j+1)\right
)\right)^{\frac{1}{2\ell}},
\]
since the odd moments disappear. Thus
\[
\bbE\|A\fky\|_\infty\leq 
\frac{1}{2}\lim_{\ell\to\infty}\left(\sum_{i=1}^k\sum_{|\alpha|=2\ell}\binom
{2\ell}{\alpha}\prod_{j=1}^n\left(|A_{i,j}|^{\alpha_j}\right)\right)^{\frac{1}
{2\ell}}=\frac{\|A\|_{\infty,\infty}}{2},
\]
where we obtained the bound of $\|A\|_{\infty,\infty}/2$ after doing the
binomial sum and taking the limit.
\end{proof}

\section{Number of complex roots}\label{sec:complexroots}

To control the number of complex roots, we will use results from complex
analysis and the probabilistic bounds from Section~\ref{sec:cond-sep-prob}.
Note that we
cannot bound the number of complex roots inside $\bbD$, because the symmetry on
our randomness model forces any bound on the number of roots in  $\bbD$ to be
of the form $\Oh(d)$. For of this, we consider a family of disks
$\{\bbD(\xi_{n,N},\rho_{n,N})\}_{n=-N}^N$, inspired by the one in
\cite{moroz2021}, where we will specify  $N$ in the sequel.
In particular,
\begin{equation}
\label{eq:Disc-center}
\xi_{n,N}=\begin{cases}
\sgn(n)\left(1-\frac{3}{4}\frac{1}{2^{|n|}}\right),&\text{if }|n|\leq N-1\\
\sgn(n)\left(1-\frac{1}{2^{N}}\right),&\text{if }|n| = N
\end{cases}
\end{equation}
\begin{equation}
\label{eq:Disc-radius}
\rho_{n,N}=\begin{cases}
\frac{3}{8}\frac{1}{2^{|n|}},&\text{if }|n|\leq N-1\\
\frac{3}{2}\frac{1}{2^{N}},&\text{if }|n| = N
\end{cases}.
\end{equation}
We will abuse notation and write $\xi_n$ and $\rho_n$ instead of $\xi_{n,N}$
and $\rho_{n,N}$ since we will not be working with different $N$'s at the same
time, but only with one $N$ which might not have a prefixed value. For this
family of disks, we will give a deterministic and a probabilistic bound for the
number of roots in their union, when $N = \lceil \lg{d} \rceil$,
\begin{equation}
\label{eq:Omega_N}
\varrho(f):=\#\Bigg\{ z\in\Omega_d:=\bigcup_{n=-\lceil\log d\rceil}^{\lceil\log
d\rceil}\bbD(\xi_{n},\rho_{n})\mid f(z)=0 \Bigg\} ,
\end{equation}
where $f\in \Pd$.
We use these bounds to estimate the number of steps of \nameref{alg:Descartes}.

\subsection{Deterministic bound}

\begin{theo}\label{theo:deterministicbopundroots}
Let $f\in\Pd$. Then
\[
\varrho(f)\leq \sum_{n=-\lceil\log d\rceil}^{\lceil\log d\rceil}\log
\frac{\enumber\|f\|_1}{|f(\xi_{n})|}.
\]
\end{theo}
\begin{lem}\label{lem:rootsinsmalldisk}
Let $f\in\Pd$, $\xi\in \bbD$, and $\rho>0$. 
If $|\xi|+2\rho<1+1/d$, then
$
	\#(\mcZ(f)\cap \bbD(\xi,\rho))
	\leq \log (\enumber\|f\|_1/|f(\xi)|)
$.
\end{lem}
\begin{proof}[Proof of Theorem~\ref{theo:deterministicbopundroots}]
We only have to apply subadditivity and Lemma~\ref{lem:rootsinsmalldisk}. Note
that the condition of the Lemma~\ref{lem:rootsinsmalldisk} holds for every disk
$\bbD(\xi_{n},\rho_{n})$ in $\Omega_d$.
\end{proof}

\begin{proof}[Proof of Lemma~\ref{lem:rootsinsmalldisk}]
We use a classic result of Titchmarsh~\cite[p.~171]{titchmarsh1939} that bounds
the number of roots in a disk. For $\delta\in(0,1)$, we have that
$
\#(\mcZ(f)\cap \bbD(\xi,\rho))\leq (\ln(1/\delta))^{-1}\ln
(\max_{z\in\bbD}|f(\xi+\rho z/\delta)|/|f(\xi)|)
$.

Take $\delta=1/2$. By our assumption, $\xi+2\rho \bbD\in (1+1/d)\bbD$, so $
\max_{z\in\bbD}|f(\xi+\rho z/\delta)|\leq \max_{z\in(1+1/d)\bbD}|f(z)|\leq
\enumber \|f\|_1
$,
since $|f(z)|\leq \enumber \|f\|_1$, for $z\in(1+1/d)\bbD$~\cite[Proposition
3.9.]{TCTcubeI-journal}.
\end{proof}

\subsection{Probabilistic bound}

\begin{theo}\label{theo:probbopundroots}
Let $\fkf\in\Pd^{\bbZ}$ be a random bit polynomial. Then for all $t\leq
\tau(\fkf)(2\lceil\lg d\rceil+1)$,
\[
\bbP\left(\varrho(\fkf)\geq t\right)
\leq 
44d^2{(2\lceil\lg d\rceil+1)}\enumber^{u(\fkf)}\enumber^{-\frac{t}{{2\lceil\lg
d\rceil+1}}}.
\]
\end{theo}
\begin{cor}\label{cor:probboundroots}
Let $\fkf\in\Pd^{\bbZ}$ be a random bit polynomial and $\ell\in \bbN$. Suppose
that $\tau(\fkf)\geq 10\ln(\enumber d)+2u(\fkf)$. Then
\[
\Big(\bbE\varrho(\fkf)^\ell\Big)^{\frac{1}{\ell}}
\leq 
2(1+\ell)(6\ln(ed)+u(\fkf))\ln(\enumber
d)+\Big(\frac{44d^{3+2\ell}\enumber^{u(\fkf)}}{2^{\tau(\fkf)}}\Big)^{\frac{1}
{\ell}}.
\]
In particular, if $\tau(\fkf)\geq (9+3\ell)\ln(\enumber d)+2u(\fkf)$, then
\[
\left(\bbE\varrho(\fkf)^\ell\right)^{\frac{1}{\ell}}\leq \Oh\left(\ell (\ln
d+u(\fkf))\ln d\right).
\]
\end{cor}

\begin{proof}[Proof of Theorem~\ref{theo:probbopundroots}]
If $\#\left(\mcZ(\fkf)\cap \Omega_d\right)\geq t$, then, by
Theorem~\ref{theo:deterministicbopundroots}, there is an $n$ such that
$
\log (\enumber\|f\|_1/|\fkf(\xi_{n})|)\geq t/(2\lceil\lg d\rceil+1)
$.
Hence
\[
\bbP\left(\varrho(\fkf)\geq t\right)\leq 
\sum_{n=-\lceil\lg d\rceil}^{\lceil\lg d\rceil} \bbP\left(\lg
\frac{\enumber\|f\|_1}{|\fkf(\xi_{n})|}
\geq \frac{t}{2\lceil\lg d\rceil+1}\right).
\]
Now, fix $x\in I$. We argue as in the proof of
Theorem~\ref{theo:reallocalconditionnumberprob}, but we consider that map
mapping $f$ to $f(x)$ instead of the map mapping $f$ to $(f(x),f'(x)/d)$, so
that our matrix $A$ takes the form
\[\begin{pmatrix}
1&x&x^{d-1}&x^d
\end{pmatrix}.\]
Note that this $A$ has $\|A\|_{\infty,\infty}\leq d+1$. So, we can apply
Proposition~\ref{prop:reallinearprojectiondiscreterandomvector} to show that
for any $s\leq 2^{\tau(\fkf)}$,
\[
\bbP\left(\enumber\|\fkf\|_1/|\fkf(x)|\geq s\right)\leq
44d^2\enumber^{u(\fkf)}/s.
\]
If $s=\enumber^{t/N}$, 
with $N = 2\lceil\lg(d) \rceil +1$, then the bound follows.
\end{proof}
\begin{proof}[Proof of Corollary~\ref{cor:probboundroots}]
In the proof of Corollary~\ref{cor:realglobalexpectations}  we only used the
fact that the tail bound is of the form $U\enumber^{-t}$ for $t\leq V$ with
$U\leq V$. We will use a similar idea in this proof. Let $0\leq U\leq V$,
$c>0$, and $\fkx\in[0,\infty)$ a random variable. If $\bbP(\fkx\geq t)\leq
\enumber^{U-s}$ for $s\leq V$, then
$\bbE(\min\{\fkx,c\})^\ell\leq U^\ell+\ell^\ell
U^{\ell-1}+\enumber^{U-V}c^\ell$.

By Theorem~\ref{theo:probbopundroots}, the random variable
$\varrho(\fkf)/(2\lceil\lg d\rceil+1)$
satisfies the conditions to be a random variable $\fkx$ with
$U=\ln(44d^2(2\lceil\lg d\rceil+1)\enumber^{u(\fkf)})\leq 4\ln(\enumber d)+\ln
(2\lceil\lg d\rceil+1)+u(\fkf)$, $V=\ln(2^{\tau(\fkf)}/(2\lceil\lg
d\rceil+1))$, and $c=\frac{d}{(2\lceil\lg d\rceil+1)}$; since the roots are at
most $d$. By our assumptions $U\leq V$,
that~concludes~the~proof.
\end{proof}

\section{The Descartes solver}
\label{sec:Descartes}

The \descartes solver is an algorithm that is based on Descartes' rule of signs.

\begin{theo}[Descartes' rule of signs]
  \label{thm:Desc-rule-of-sign}
  The number of sign variations in the coefficients' list of a
  polynomial $f = \sum_{i=0}^df_i X^i \in \Pd$ equals the number of
  positive real roots (counting multiplicities) of $f$, say $r$,
  plus an even number; that is
  $r \equiv \var(f) \mod 2$.\eproof
\end{theo}

\begin{algorithm2e}[t]
  \scriptsize \dontprintsemicolon \linesnumbered
  \SetFuncSty{textsc} \SetKw{RET}{{\sc return}} \SetKw{OUT}{{\sc output \ }}
  \SetVline \KwIn{A square-free polynomial $f\in \Pd^\bbZ$}
\KwOut{A list, $S$, of isolating intervals for the real roots of $f$ in $J_0 =
(-1, 1)$}

  \BlankLine

  $J_0 \leftarrow (-1, 1)$, 
  $S \leftarrow \emptyset,\, Q \leftarrow \emptyset$,
  $Q \leftarrow \FuncSty{push}( {J_0})$ \;

  \While{ $Q \neq \emptyset$}{
    \nllabel{alg:Subdivision-while-loop}

    ${J} \leftarrow \FuncSty{pop}( Q)$
    \tcp*{Assume that $J = (a, b)$.}
    $V \leftarrow \var(f, J)$ \;

    \Switch{ $V$ }{

      \lCase{ $V = 0$ }{ \KwSty{continue}\; }
      \lCase{ $V = 1$ }{ $S \leftarrow \FuncSty{ add}( {I})$ \; }
      \Case{ $V > 1$ } {
        $m \gets \frac{a+b}{2}$ \;
        \lIf{$f(m) = 0$}{  $S \leftarrow \FuncSty{ add}( {[m, m]})$ \; }
        $J_L \gets [a, m]$ ;   $J_R \gets [m, b]$ \;

        $Q \leftarrow \FuncSty{push}( Q, {J_L} )$,
        $Q \leftarrow \FuncSty{push}( Q, {J_R} )$ \;
      }
    }
  }
  \RET $S$ \;
  \caption{ $\func{Descartes}(f)$}
  \label{alg:Descartes}
\end{algorithm2e}

In general, Theorem~\ref{thm:Desc-rule-of-sign} provides an
overestimation on the number of positive real roots. It counts exactly
when the number of sign variations is 0 or 1
and if the polynomial is hyperbolic, that is it has \emph{only} real roots.
%
To count the real roots of $f$ in an interval
$J = (a, b)$
we use the transformation $x \mapsto \frac{ax + b}{x + 1}$
that maps $J$ to $(0, \infty)$.
Then
\[\var(f, J) := \var( (X + 1)^d f(\tfrac{aX + b}{X + 1}) )\]
bounds the number of real roots of $f$ in $I=J$.

Therefore, to isolate the real roots of $f$ in an interval, say  $J_0 = (-1,
1)$,
we count (actually bound) the number of roots of $f$ in $J_0$ using $V =
\var(f, J_0)$.
If $V = 0$, then we discard the interval.
If $V = 1$, then we add $J_0$ to the list of isolating intervals.
If $V > 1$, then we subdivide the interval to two intervals $J_L$ and $J_R$
and we repeat the process.
The pseudo-code of \descartes appears in Algorithm~\ref{alg:Descartes}.

The recursive process of the \descartes defines a binary tree. Every node of the
tree corresponds to an interval. The root corresponds to the initial interval
$J_0=(-1, 1)$. If a node corresponds to an interval $J=(a,b)$, then its
children correspond to the open left and right half intervals of $J$, that is
$J_L = (a, \xmid(J))$ and $J_R = (\xmid(J), b)$ respectively.
The internal nodes of the tree correspond to intervals $J$, such that
$\var(f,J) \geq 2$.
The leafs correspond to intervals that contain 0 or 1 real roots of $f$. 
Overall, the number of nodes of the tree correspond to the number of
steps, i.e., subdivisions, that the algorithm performs.  We control the number
of nodes by controlling
the depth of tree and the width of every layer. Hence, to obtain the
final complexity estimate it suffices to multiple the number of steps (width
times height) with the
worst case cost of each step.

The following proposition helps to control the cost of each step. Note that at
each step, we do changes of variables to obtain the desired polynomial to
perform the sign count.

\begin{prop}\label{prop:polymaps}
  \label{prop:poly-maps}
  Let $f=\sum_{i=0}^df_iX^i\in\Pd^\bbZ$ of bit-size $\tau$.
  \begin{itemize}[leftmargin=*]
    \item The reciprocal transformation is
      $R(f): = X^{d} f(\tfrac{1}{X}) = \sum_{k=0}^{d}f_{d-k}X^k$. Its cost is
      $\OB(1)$ and it does not alter neither the degree nor the bit-size of the
      polynomial.
    \item The homothetic transformation of $f$ by $2^k$, for a positive integer
      $k$, is
$H_k(f) = 2^{d k} f(\tfrac{X}{2^k}) = \sum_{i=0}^d{2^{k(d-i)} f_i \, X^{i}}$.
      It costs $\OB(d \, \mu(\tau + d k)) = \sOB(d\tau + d^2 k)$ and the
      resulting polynomial has bit-size $\OO(\tau + d k)$. Notice that
      $H_{-k} = R H_{k} R$.

\item The Taylor shift of $f$ by in integer $c$ is $T_c(f) = f(x+c) =
\sum_{k=0}^d{a_k x^k}$,
where $a_i = \sum_{j=i}^{d}\binom{j}{i}f_j c^{j-i}$ for $0 \leq i \leq d$. It
costs $\OB(\mu(d^2\sigma + d\tau)\lg{d}) = \sOB(d^2\sigma + d\tau)$
     \cite[Corollary~2.5]{vzGGer}, where $\sigma$ is the bit-size of $c$. The
     resulting polynomial has bit-size $\OO(\tau + d\sigma)$.\eproof
  \end{itemize}
\end{prop}

\begin{remark}
  There is no restriction on working with open intervals since we
consider an integer polynomial and we can always evaluate it at the endpoints.
  Also to isolate all the real roots of $f$ it suffices to have a routine to
  isolate the real roots in $(-1, 1)$. Using the map $x \mapsto 1/x$ we can
  isolate the roots in $(-\infty, -1)$ and  $(1, \infty)$.
\end{remark}
\subsection{Bounds on the number of sign variations}
\label{sec:Desc-termination}

For this subsection we consider $f=\sum_{i=0}^df_iX^i\in\Pd$
to be a polynomial with real coefficients, not necessarily integers.
To establish the termination and estimate the bit complexity of \descartes
we need to introduce the Obreshkoff area and lens.
Our presentation follows closely \cite{sm-anewdsc,KM-newDesc-06,emt-lncs-2006}.

Consider $0 \leq \varrho \leq d$ and a real open interval $J = (a, b)$. The
\emph{Obreshkoff discs} $\uObD_{\varrho} $ and $\dObD_{\varrho}$ are discs the
boundaries of which go through the endpoints of $J$. Their centers are above,
respectively below, $J$ and they form an angle $\varphi = \frac{\pi}{\varrho+2}$
with the endpoints of $I$. Its diameter is
$\wid(J)/ \sin(\frac{\pi}{\varrho+2})$.

The \emph{Obreshkoff area} is $\ObA_{\varrho}(J) =
\mathsf{interior}(\uObD_{\varrho} \,\cup\, \dObD_{\varrho})$;
it appears with grey color in Fig.~\ref{fig:Obreshkoff}.
The \emph{Obreshkoff lens} is $\ObL_{\varrho}(J) =
\mathsf{interior}(\uObD_{\varrho} \,\cap\, \dObD_{\varrho})$;
it appears in light-grey color in Fig.~\ref{fig:Obreshkoff}.
If it is clear from the context, then we omit $I$ and we write $\ObA_{\varrho}$
and  $\ObL_{\varrho}$,
instead of $\ObA_{\varrho}(J)$ and $\ObL_{\varrho}(J)$.
It holds that
$\ObL_d \subset \ObL_{d-1} \subset \cdots \subset \ObL_1 \subset \ObL_0$
and
$\ObA_0 \subset \ObA_1 \subset \cdots \subset \ObA_{d-1} \subset \ObA_d$.

\begin{figure}[t]
  \centering
  \includegraphics[scale=0.4]{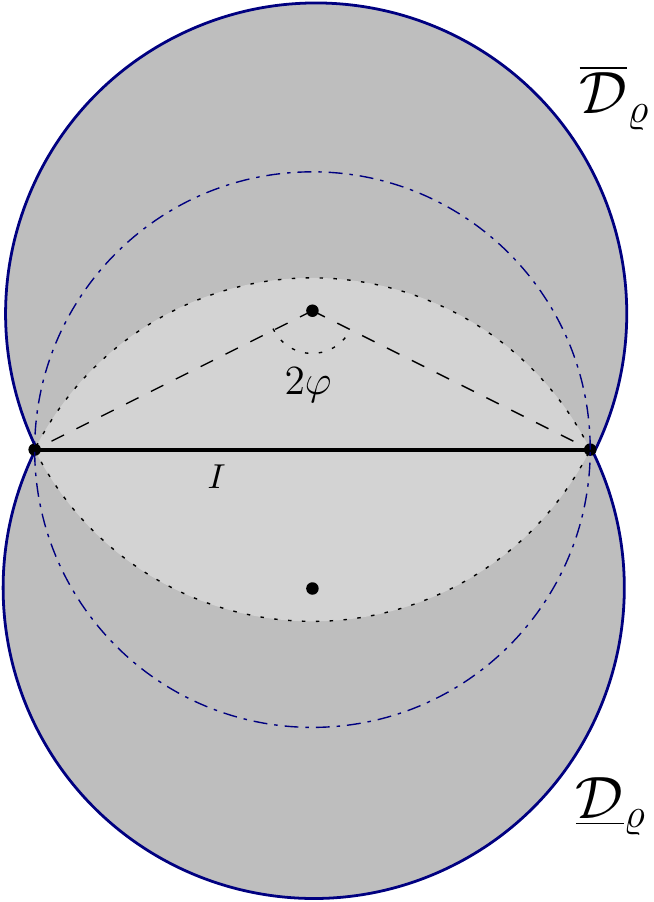}
\caption{Obreshkoff discs, lens (light grey), and area  (light grey and grey)
for an interval $I$.}
  \label{fig:Obreshkoff}
\end{figure}

The following theorem shows the role of the number of complex roots in the
control of the number of variation signs.

\begin{theo}[\cite{Obreshkoff-book}]
  \label{thm:Obr-signs}
  Consider $f\in\Pd$ 
  and real open interval $J=(a,b)$.
If the Obreshkoff lens $\ObL_{d - k}$ contains at least $k$ roots (counted with
multiplicity) of $f$,
  then $k \leq \var(f,J)$.
If the Obreshkoff area $\ObA_{k}$ contains at most $k$ roots (counted with
multiplicity) of $f$,
  then $\var(f,J) \leq k$.
  Especially
  \begin{equation*}\tag*{\qed}
    \# \{ \text{roots of }f\text{ in } \ObL_d \}
    \le  \var(f,J) \le
    \# \{ \text{roots of }f\text{ in } \ObA_d \}.
  \end{equation*}
\end{theo}

This theorem together with the subadditive property
of Descartes' rule of signs (Thm.~\ref{thm:desc-subadd})
shows that the number of complex roots in the Obreshkoff areas
controls the width of the subdivision tree of \descartes.

\begin{theo}
  \label{thm:desc-subadd}
  Consider a real polynomial $f\in\Pd$.
Let $J$ be a real interval and $J_1, \dots, J_n$ be disjoint open subintervals
of $J$.
  Then, it holds $\sum_{i=1}^n \var(f, J_i) \leq \var(f,J)$.\eproof
\end{theo}

Finally, to control the depth of the subdivision tree of \descartes we use the
one and two circle theorem~\cite{AleGalu-vincent-98,KM-newDesc-06}.
We present a variant
based on the $\varepsilon$-real separation of $f$, $\Delta_{\varepsilon}\uR(f)$
(Definition~\ref{defi:realseparation}).

\begin{theo}\label{thm:sepdepthbound}
Let $f\in\Pd$, an interval $J\subseteq (-1,1)$ and $\varepsilon>0$. If
\[
2 \, \wid(J)\leq \min\{\Delta_{\varepsilon}\uR(f),\varepsilon\} ,
\]
then $\var(f, J)=0$ (and $J$ does not contain any real root), or $\var(f, J)=1$
(and $J$ contains exactly one real root).
\end{theo}
\begin{proof}
The proof follows the same application of the one and two circle theorems as in
the proof of \cite[Proposition~6.4]{TCTcubeI-journal}.
\end{proof}

\subsection{Complexity estimates for \descartes}
\label{sec:Descartes-complexity-exp}
We give a high-level overview of the proof ideas of this section before going
into technical details. The process of \descartes corresponds to a binary tree
and we control its depth using the real condition number and
Theorems~\ref{theo:condbasedseparation} and~\ref{thm:sepdepthbound}.
To bound the width of the \descartes' tree we use the Obreskoff areas and the
number of complex roots in them (Theorem~\ref{thm:Obr-signs}).
By combining these two bounds, we control the size of the tree and so we obtain
an instance-based complexity estimate. To turn this instance-based complexity
estimate into an expected one, we use
Theorems~\ref{theo:realglobalconditionnumberprob}
and~\ref{theo:probbopundroots} (and their
Corollaries~\ref{cor:realglobalexpectations} and \ref{cor:probboundroots}).

\subsubsection{Instance-based estimates}
\begin{theo}
 \label{thm:Descartes-steps}
If $f\in\Pd^\bbZ$, then,  using  
\descartes,  the number of subdivision steps to isolate the real roots in $I =
(-1, 1)$
	is 	 
	\[
	\sO(\varrho(f)^2\lg(\condR(f)).
	\] 
The bit complexity of the algorithm is
  \[
    \sOB(d \tau \varrho(f)^2 \lg \condR(f)+ d^2 \varrho(f)^2 \lg^2\condR(f)).
  \]
\end{theo}

Recall that $\condR(f)$ appears in~\eqref{eq:condR} and $\varrho(f)$
in~\eqref{eq:Omega_N}.

\begin{figure}[!tbp]
  \centering
  \begin{subfigure}[b]{0.2\textwidth}
    \centering
    \includegraphics[scale=0.55]{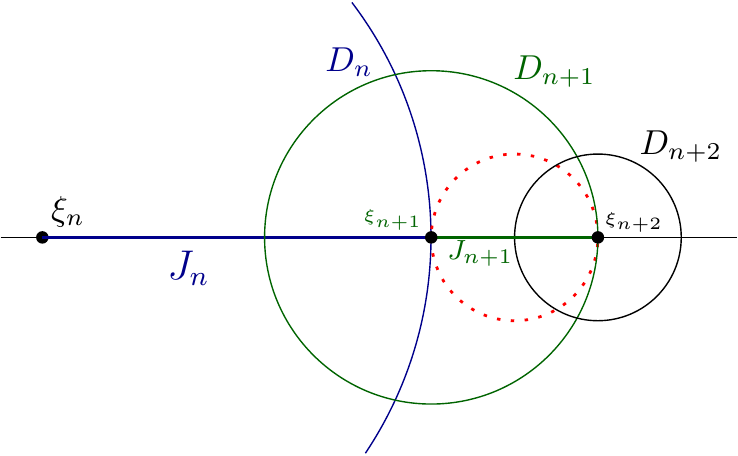}
    
        \label{fig:covering-discs}
  \end{subfigure}%
  \hfill
  \begin{subfigure}[b]{0.20\textwidth}
    \centering
    \includegraphics[scale=0.5]{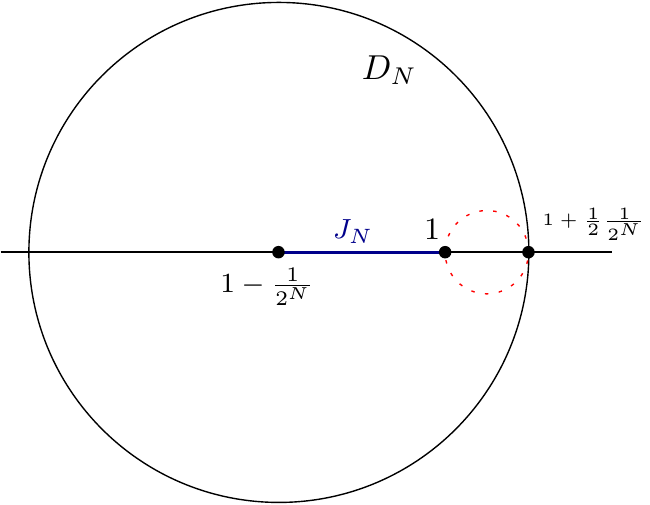}
    \label{fig:covering-last-disc}
  \end{subfigure}
  \label{fig:covering-discs-all}
\caption{Covering discs of the interval $I = (0, 1)$.\\(left) Three covering
discs, $D_n, D_{n+1}, D_{n+2}$.\\(right) The (red) dotted circle is the
auxiliary disc that we  ensure  is contained in $D_{n+1} \setminus D_{n}$.}
\end{figure}

\begin{proof}
We consider the number of steps to isolate the real roots in $I = (-1, 1)$. Let
$N=\ceil{\log d}$ and $\varrho=\varrho(f)$ the number of complex roots in
$\Omega_d$. Recall that $\Omega_d$ is the union of the discs
$D_n:=\bbD(\xi_n,\rho_n):=\xi_n+\rho_n\bbD$, where $\abs{n} \leq N$; see
\eqref{eq:Disc-center} and \eqref{eq:Disc-radius} for the concrete formulas,
and that it contains the interval $I$.

The discs partition $I$ into the $2N+1$ subintervals $J_n := [\xi_n,
\xi_{n+1}]$ (or $J_n := [\xi_n, \xi_{n-1}]$ if $n \leq 0$). Note that $J_n$ is
the union of $3$ intervals of size $1/2^{n+3}$. Because of this, there is a
binary subdivision tree of $I$ of size $\Oh(\lg^2 d)$ such that every of its
intervals is contained in some $J_n$. Thus, if we bound the width of the
subdivision tree of \descartes starting at each $J_n$ by $w$, then the width of
the subdivision tree of \descartes starting at $I$ is bounded by $\Oh(w\lg^2
d+\lg^2 d)$.

We focus on intervals $J_n$ for $n \geq 0$; similar arguments apply for $n \geq
0$. We consider two cases: $n<N$ and $n=N$.

\noindent
\emph{Case $n<N$.} It holds $\wid(J_n) = \rho_n = 3/2^{n+3}$. For each $J_n$,
assume that we perform a number of subdivision steps
to obtain intervals, say $J_{n, \ell}$, with
$\wid(J_{n,\ell}) = 2^{-\ell}$. We
choose $\ell$ so that the corresponding Obreshkoff areas,
$\ObA_{\varrho}(J_{n,\ell})$, are inside $\Omega_d$. In particular, we ensure
that the Obreshkoff areas related to $J_{n,\ell}$ lie in $D_{n+1}$.
  
The diameter of the Obreshkoff discs,
$\uObD_{\varrho}(J_{n,\ell})$ and $\dObD_{\varrho}(J_{n,\ell})$, 
is $\wid(J_{n,\ell})/\sin\tfrac{\pi}{\varrho + 2}$.
For every $\ObA_{\varrho}(J_{n,\ell})$ to be in $D_{n+1}$ and hence inside
$\Omega_d$,
it suffices that a disc with diameter 
$2 \,\wid(J_{n,\ell})/\sin\tfrac{\pi}{\varrho + 2}$,
that has its center in the interval $[\xi_n, \xi_{n+1}]$
and touches the right endpoint of $J_n$, to be
inside $D_{n+1} \setminus D_n$.
This is the worst case scenario: a disc big enough that contains
$\ObA_{\varrho}(J_{n,\ell})$ and lies $D_{n+1}$. This auxiliary disc is the
dotted (red) disc in Fig.~2~(left).
It should be that 
  \[
    2\, \wid(J_{n,\ell})/ {\sin\tfrac{\pi}{\varrho + 2}} \leq 
    2\,\rho_{n+1} = 3/2^{n+3}.
  \]
Taking into account that  $\wid(J_{n,\ell}) = 2^{-\ell}$ and
  \[
    \sin\tfrac{\pi}{\varrho + 2}
    > \sin\tfrac{1}{\varrho}
    \geq {\tfrac{1}{\varrho}} / {\sqrt{1 + \tfrac{1}{\varrho^2}}}
    \geq \tfrac{1}{2 \varrho} ,
    \]
we deduce
$
	2^{-\ell +1} 2 \varrho \leq 3/2^{n+3}
$
and so
$
	\ell \geq \lg \frac{2^{n+5} \varrho}{3}.
$

Hence, $\wid(J_{n, \ell}) = 3/(2^{n+5} \varrho)$
and so $J_n$ is partitioned to at most
$\frac{\wid(J_n)}{\wid(J_{n,\ell})} = 4 \varrho$
(sub)intervals.
So, during the subdivision process, starting from (each) $J_n$, we obtain the
intervals
$J_{n,\ell}$ after performing
 at most $8 \varrho$ subdivision steps (this is the size of the
 complete binary tree starting from $J_n$). 
 To say it differently, the subdivision tree
 that has $J_n$ as its root and the intervals $J_{n, \ell}$ 
as leaves has depth $\ell = \lceil \lg(4 \varrho) \rceil$. The same hold for
$J_{N-1}$ because
$\rho_n \leq \rho_{N}$, for all $0 \leq n \leq N-1$. 

Thus, the width of the tree starting at $J_n$ is at most $\Oh(\varrho^2)$,
because we have $\OO(\varrho)$ subintervals $J_{n,\ell}$
and for each $\var(f, J_{n,\ell}) \leq \varrho$.

\noindent
\emph{Case $n=N$.} 
Now $\wid(J_N) = 3/2^{N+1}$.
We need a slightly different argument to account for the number of subdivision
steps
for the last disc $D_N$. 
To this disc we assign the interval $J_N = [1 - 1/2^N, 1]$ with $\wid(J_N) =
1/2^N$;
see Figure~2.

We need to obtain small enough intervals $J_{N, \ell}$ of width $1/2^{\ell}$
so that corresponding Obreskoff areas,  $\ObA_{\varrho}(J_{N,\ell})$, to be
inside $D_N$.
So, we require that an auxiliary disc of diameter 
$2 \,\wid(J_{N,\ell})/\sin\tfrac{\pi}{\varrho + 2}$,
that has ts center in the interval $[1, 1/2^{N+1}]$
and touches 1 to be
inside $D_{N}$; actually inside $D_N \cap \{x \geq 1\}$; see Figure~2.
And so 
 \[
    2\, \wid(J_{N,\ell}) / {\sin\tfrac{\pi}{\varrho + 2}} \leq 
    \rho_{n+1} = 1/2^{N+1}.
  \]
This leads to $\ell \geq \lg(\varrho \, 2^{N+3})$.
Working as previously, we estimate that the number of subdivisions we perform
to obtain the interval $J_{N, \ell}$ is $8\varrho$.
Also repeating the previous arguments, the width of the tree of \Descartes
starting at $J_N$ is at most $\Oh(\varrho^2)$.

By combining all the previous estimates, we conclude that the subdivision tree
of \descartes has width $\Oh(\varrho^2\lg^2 d+\lg^2 d)$.

To bound the depth of the subdivision tree of \descartes, consider an interval
$J_\ell$ of width $1/2^{\ell}$ obtained after $\ell+1$ subdivisions. By
theorem~\ref{thm:sepdepthbound}, we can guarantee termination if for some
$\varepsilon>0$,
\[
1/2^{\ell - 1}\leq \min\{\Delta_{\varepsilon}\uR(f),\varepsilon\}.
\]
Fix $\varepsilon=1/(\enumber d\condR(f))$. Then, by
Theorem~\ref{theo:condbasedseparation}, it suffices to hold
\[
\ell\geq 1+\lg(12 d\condR(f)).
\]
Hence, the depth of the subdivision tree is at most $\Oh(\lg(d\condR(f)))$.

Therefore, since the subdivision tree of \descartes has width
$\Oh(\varrho^2\log d+\log^2 d)$ and depth $\Oh(\lg(d\condR(f)))$, the size
bound follows. For the bit complexity, by \cite{ESY:descartes}, see also
\cite{KM-newDesc-06,sm-anewdsc,Sagraloff-approxDesc-14,emt-lncs-2006} and
Proposition~\ref{prop:poly-maps},
the worst case cost of each step of \descartes is $\sOB(d \tau + d^2 \delta)$,
where $\delta$ is the logarithm of the highest bitsize that we compute with, or
equivalently the depth of the subdivision tree. In our case, $\delta =
\OO(\lg(d \condR(f))$.
\end{proof}

\subsubsection{Expected complexity estimates}

\begin{theo}
  \label{thm:Descartes-complexity-exp}
Let $\fkf \in \Pd^\bbZ$ be a random bit polynomial with $\tau(\fkf) \geq
\Omega(\lg{d} +u(\fkf))$. Then, using \descartes, the expected number of
subdivision steps to isolate the real roots in $I = (-1, 1)$	is
	\[
	\sO((1+u(\fkf))^3).
	\] 
	The expected bit complexity of \descartes is
  \[
    \sOB(d \, \tau(\fkf) (1+u(\fkf))^3 + d^2(1+u(\fkf))^4 ).
  \]
If $\fkf$ is a uniform random bit polynomial of bitsize $\tau$ and $\tau =
\Omega(\lg{d} +u(\fkf))$,en the expected number of subdivision steps to
isolate the real rin $I = (-1, 1)$	is
	$
	\sO(1)
	$ 
and the expected bit  complexity becomes
  \[
    \sOB(d \tau  + d^2 ).
  \]
\end{theo}
\begin{proof}
We only bound the number of bit operations; the bound for the number of steps
is analogous. By Theorem~\ref{thm:Descartes-steps} and the worst-case bound
$\sOB(d^4 \tau^2)$ for \descartes \cite{ESY:descartes}, the bit complexity of
\descartes at $\fkf$ is at most
\[
\sOB\left(\min\{d \tau(\fkf) \varrho(\fkf)^2 \lg \condR(\fkf)+ d^2
\varrho(\fkf)^2 \lg^2\condR(\fkf)),d^4\tau(\fkf)^2\}\right),
\]
that in turn we can bound by 
\begin{multline*}
\sOB\left(d \tau(\fkf)\varrho(\fkf)^2\min\{\lg
\condR(\fkf),d^3\tau(\fkf)\}\right.\\\left.+  d^2 \varrho(\fkf)^2\min \{
\lg\condR(\fkf),d^2\tau(\fkf)^2 ) \}\}\right).
\end{multline*}
Now, we take expectations, and, by linearity, we only need to bound
\[
\bbE\,\varrho(\fkf)^2 \min\{\lg \condR(\fkf),d^3\tau(\fkf)\}\text{ and
}\bbE\,\varrho(\fkf)^2 \left(\min \{\lg\condR(\fkf),d^2\tau(\fkf)^2\}\right)^2 .
\]
Let us show how to bound the first, because the second one is the same. By the
Cauchy-Bunyakovsky-Schwarz inequality, $\bbE\,\varrho(\fkf)^2 \min\{\lg
\condR(\fkf),d^3\tau(\fkf)\}$ is bounded by
\[
\sqrt{\bbE\,\varrho(\fkf)^4}\sqrt{\bbE\,\left(\min\{\lg
\condR(\fkf),d^3\tau(\fkf)\}\right)^2}.
\]
Finally, Corollaries \ref{cor:realglobalexpectations} and
\ref{cor:probboundroots} give the estimate.
Note that $\tau(\fkf) \geq \Omega(\lg{d} +u(\fkf))$ implies $\tau(\fkf) \geq
\Omega(\lg{d} +u(\fkf)+\ln c)$ (for the worst-case separation bound $c$
\cite{Dav:TR:85}) so we can apply Corollary~\ref{cor:realglobalexpectations}.
\end{proof}

\newpage
\bibliographystyle{abbrv}
\bibliography{biblio.bib,biblio2.bib}

\end{document}